\documentclass[reqno]{amsart}
\usepackage{bm}
\usepackage{amsmath}
\usepackage{amsfonts,amssymb}
\usepackage[dvips]{graphicx}
\usepackage{epsfig}
\usepackage{color}
\usepackage{ulem}
\newtheorem{theorem}{Theorem}[section]
\newtheorem{remark}{Remark}[section]
\newtheorem{lemma}[theorem]{Lemma}
\newtheorem{definition}{Definition}[section]

\newtheorem{proposition}[theorem]{Proposition}
\def\labelenumi{\theenumi}

\numberwithin{equation}{section}
\begin{document}

\title[Hotelling's model]
{Necessary and sufficient condition for equilibrium of the Hotelling model}
\author{Satoshi Hayashi${}^{\dagger}$}
\address{${}^{\dagger\ddagger}$Department of Mathematics Education, 
	Faculty of Education, Gifu University, 1-1 Yanagido, Gifu
	Gifu 501-1193 Japan.}
\email{${}^{\dagger}$x1131023@edu.gifu-u.ac.jp}

\author{Naoki Tsuge${}^{\ddagger}$}

\email{${}^{\ddagger}$tuge@gifu-u.ac.jp}
\thanks{
	N. Tsuge's research is partially supported by Grant-in-Aid for Scientific 
	Research (C) 17K05315, Japan.
}


\date{}

\maketitle

\begin{abstract}
	We study a model of vendors competing to sell a homogeneous product to customers spread 
	evenly along a linear city. This model is based on Hotelling's celebrated paper in 1929.  Our aim in this paper is to present a necessary and sufficient condition for the equilibrium. This yields a representation for the equilibrium. To achieve this, we first formulate the model mathematically. Next, we prove that the condition holds if and only if 
	vendors are equilibrium. 
\end{abstract}


\section{Introduction}
We study a model in which a linear city of length $1$ on a line and 
customers are uniformly distributed with density $1$ along this interval. 
We consider $n$ vendors moving on this line. Let the location of 
the vendor $k\quad(k=1,2,3,\ldots,n)$ be $x_k\in[0,1]$.  We assume that $x_1\leq x_2\leq \cdots\leq x_n$ and  denote the location of $n$ vendors $(x_1,x_2,x_3,\ldots,x_n)$ by ${\bm x}$. Since we study the 
competition between vendors, we consider $n\geq2$ in particular.
The price of one unit of product for each vendor is 
identical. Moreover, we assume the following. 
\vspace*{1ex}

If there exist $l\quad(l=1,2,3,\ldots,n)$ vendors nearest to 
a customer, the customer  
 purchases $1/l$ unit of product per unit of time from each of the $l$ vendors respectively. \vspace*{-1ex}

Every vendor then seeks a location to maximize his profit.

We then represent the profit of  vendor $k$ per unit of time by a mathematical notation.
Given a vector ${\bm \xi}=(\xi_1,\xi_2,\ldots,\xi_n)\in[0,1]^n$ and $0\leq y\leq1$, 
we define a set $\displaystyle S({\bm \xi},y)=\{j\in\{1,2,3,\ldots,n\}:|\xi_j-y|=\min_{i}|\xi_i-y|\}$.
By using a density function 
\begin{align*}
\rho_k({\bm \xi},y)=
\begin{cases}
\displaystyle 0&\displaystyle \text{if}\quad|\xi_k-y|>\min_{i}|\xi_i-y|\\
\displaystyle \frac{1}{|S({\bm \xi},y)|}&\displaystyle \text{if}\quad|\xi_k-y|=\min_{i}|\xi_i-y|
\end{cases},
\end{align*}
we define 
\begin{align}
f_k({\bm \xi})=\int^1_0\rho_k({\bm \xi},y)dy,
\end{align}
where $|A|$ represents a number of elements in a set $A$. We call $f_k({\bm x})$ the profit of vendor $k$ per unit of time for a location ${\bm x}$. We then define equilibrium as follows.
\begin{definition}
A location ${\bm x}^*=(x^*_1,x^*_2,x^*_3,\ldots,x^*_n)\in[0,1]^n$ is called equilibrium, if 
\begin{align}
f_k({\bm x}^*)\geq f_k(x^*_1,x^*_2,x^*_3,\ldots,x^*_{k-1},x_k,x^*_{k+1},\ldots,x^*_n)
\end{align}holds for any $k\in\{1,2,3,\ldots,n\}$ and $x_k\in[0,1]$.
\end{definition}

We review the known results. The present model is based on Hotelling's model in \cite{H}. Although we 
consider homogeneous vendors, Hotelling did heterogeneous vendors. In \cite[Chaper 10]{A}, Alonso, W. treated with the same model as our problem for two vendors. He introduced this model 
as the competition between two vendors of ice cream along a beach.
In \cite{EL2},  the model for $n$ vendors was studied. Furthermore, Eaton, B. C. and Lipsey, R. G. investigated a necessary and sufficient 
condition for equilibrium. More precisely, they claimed that (1.i) and (1.ii) in \cite[p29]{EL2} if and only if $n$ vendors are equilibrium (see also [p9]\cite{EL1}).
Although this is an interesting approach from the point of mathematical view, unfortunately, it seems that (1.i) and (1.ii) are not sufficient conditions. 
In fact, we consider a location 
\begin{align*}
{\bm x}=
\left(\frac{1}{10},\frac{1}{10},\frac{3}{10},\frac{3}{10},\frac{7}{10},\frac{7}{10},\frac{9}{10},\frac{9}{10}\right).
\end{align*}
 Then, we find that 
$\displaystyle f_1({\bm x})=f_2({\bm x})=f_7({\bm x})=f_8({\bm x})=1/{10},\;
f_3({\bm x})=f_4({\bm x})=f_5({\bm x})=f_6({\bm x})=3/{20}$. 
Therefore, 
this example satisfies (1.i) and (1.ii). \begin{remark}\normalfont
In \cite{EL2}, vendors in our problem are called firms and $f_k({\bm x})$ seems to be called a market of firm $k$. 
In addition, we regard two pairs of peripheral firms in (1.ii) of \cite{EL2} as firms $1,2$ and $9,10$. 
\end{remark}
On the other hand, if $x_3$ moves from $3/10$ to 
$\displaystyle 1/2$, $x_3$ can obtain a profit $\displaystyle 2/{10}$ more 
than the original one $\displaystyle 3/{20}$.
In addition, 
the definition of their terminologies seems not to be clear, such as market, peripheral, equilibrium, etc. 
Therefore, our goal in this paper is to formulate this model mathematically and present a revised necessary and sufficient for equilibrium.

For convenience, we set $x_0=0,\;x_{n+1}=1$ and denote a interval $[x_k,x_{k+1}]$ by $I_k\;\;(k=0,1,2,\ldots,n)$. Then our main theorem is as follows.
\begin{theorem}\label{thm:main}${}$
\begin{enumerate}
\item $n=2$\vspace*{-2ex}
\begin{align}
{\bm x}=\left(\frac12,\frac12\right)\text{ is a unique equilibrium.}\label{eqn:main3}
\end{align}
\item $n=3$\vspace*{-2ex}
\begin{align}
\text{There exits no equilibrium.}\label{eqn:main4}
\end{align} 
	\item $n\geq4$

	${\bm x}$ is equilibrium, if and only if the following conditions \eqref{eqn:main1} and \eqref{eqn:main2} hold.
	\begin{align}
	&\begin{array}{ll}
	&|I_0|=|I_n|>0\;(n\geq 2),\quad|I_1|=|I_{n-1}|=0\quad(n\geq 4),\\
	&|I_0|:|I_2|=1:2,\quad|I_{n-2}|:|I_n|=2:1\quad(n\geq 4),
	\end{array}
	\label{eqn:main1}\\
	&\hspace*{3.5ex}|I_j|\leq2|I_0|\quad(0\leq j\leq n),\quad2|I_0|\leq|I_k|+|I_{k+1}|\quad(1\leq k\leq n-2).
	\label{eqn:main2}
	\end{align}
\end{enumerate}
\end{theorem}

\section{Preliminary}
In this section, we prepare some lemmas and a proposition to prove our main theorem in a next section. 
We first consider the profit of $i$ vendors which locate at one point. We have the following lemma. 
\begin{lemma}
	We consider a location ${\bm x}=(x_1,x_2,x_3,\ldots,x_n)$ with $x_1\leq x_2\leq x_3\leq\dots\leq x_n$.
	We assume that $x_l<x_{l+1}=\cdots = x_k = \cdots = x_{l+i}<x_{l+i+1}\quad(n\geq2,\;l\geq0,\;l+i+1\leq n+1,\;1\leq i\leq n)$.
	\begin{enumerate}
		\item If $x_l\ne x_0$ and $x_{l+i+1}\ne x_{n+1}$, $\displaystyle f_k({\bm x})=\frac1{2i}(|I_l|+|I_{l+i}|)$.
		\item If $x_l=x_0$ and $x_{l+i+1}\ne x_{n+1}$, $\displaystyle f_k({\bm x})=\frac1{i}\left(|I_l|+\frac12|I_{l+i}|\right)$.
		\item If $x_l\ne x_0$ and $x_{l+i+1}=x_{n+1}$, $\displaystyle f_k({\bm x})=\frac1{i}\left(\frac12|I_l|+|I_{l+i}|\right)$.
		\item If $x_l=x_0$ and $x_{l+i+1}=x_{n+1}$, $\displaystyle f_k({\bm x})=\frac1{n}$. 
	\end{enumerate}	
\end{lemma}\begin{proof}${}$

	{\it Proof of $(\mathrm{i})$}

We have $f_k({\bm x})=\dfrac{1}{i}\left(\dfrac{1}{2}|I_l|+\dfrac{1}{2}|I_{l+i}|\right)=\dfrac{1}{2i}\left(|I_l|+|I_{l+i}|\right)$.

	{\it Proof of $(\mathrm{ii})$}

	We have $f_k({\bm x})=\dfrac{1}{i}\left(|I_l|+\dfrac{1}{2}|I_{l+i}|\right)=\dfrac{1}{i}\left(|I_l|+\dfrac{1}{2}|I_{l+i}|\right)$.

	{\it Proof of $(\mathrm{iii})$}\\
	
	We have $f_k({\bm x})=\dfrac{1}{i}\left(\dfrac{1}{2}|I_{l}|+|I_{l+i}|\right)=\dfrac{1}{i}\left(\dfrac{1}{2}|I_l|+|I_{l+i}|\right)$.
		
	{\it Proof of $(\mathrm{iv})$}

	
	We have $f_k({\bm x})=\dfrac{1}{n}\left(|I_0|+|I_n|\right)=\dfrac{1}{n}$.
	
\end{proof}

Next, the following proposition play an important role.

\begin{proposition}\label{pro:1}
	If the location  of $n$ vendors ${\bm x}=(x_1,x_2,x_3,\ldots,x_n)$ with $x_1\leq x_2\leq x_3\leq\dots\leq x_n\quad(n\geq2)$ is equilibrium,  
	the following holds.
	\begin{align}
	&x_1\ne0\text{ and }x_n\ne1.
	\label{eqn:pro1}\\
	&\text{No more than 2 vendors can occupy a location.}
	\label{eqn:pro2}\\
	&x_1=x_2\text{ and }x_{n-1}=x_n.
	\label{eqn:pro3} 
	\end{align}
\end{proposition}

\begin{proof}
${}$

	{\it Proof of \eqref{eqn:pro1}}	\\
	If $x_1=0$, we show that $ {\bm x} $ is not equilibrium.\\
	\begin{enumerate}
		\renewcommand{\labelenumi}{(\roman{enumi})}
		\item $x_1=0$ and $x_2\neq0$
		We notice that $f_1({\bm x})=\dfrac{1}{2}|I_1|$. Setting $x_1'=\dfrac{1}{2}|I_1|$, we then have
		\begin{align*}
		f_1(x_1',x_2,\cdots,x_n)&=|[0,x_1']|+\dfrac{1}{2}|[x_1',x_2]|
		>|[0,x_1']|
		>\dfrac{1}{2}|I_1|=f_1({\bm x}),
		\end{align*}
		where $|I|$ represents the length of a interval $I$.
		
		\item $x_1=\cdots=x_i=0$ and $x_{i+1}\neq0$  $(2\leq i\leq n-1)$
		

We notice that $f_1({\bm x})=\dfrac{1}{2i}|I_i|$. Setting $x_1'\in{(0,x_{i+1})}$, we have
		\begin{align*}
		f_1(x_1',x_2,\cdots,x_n)=\dfrac{1}{2}(|[0,x_1']|+|[x_1',x_{i+1}]|)
		=\dfrac{1}{2}|[0,x_{i+1}]|
		=\dfrac{1}{2}|I_i|
		>\dfrac{1}{2i}|I_i|=f_1({\bm x}).
		\end{align*}
		
\item $x_1=\cdots=x_n=0$  $(n\geq2)$

We notice that $f_1({\bm x})=\dfrac{1}{n}$. Setting $x_1'=\dfrac{1}{2}$, we have
		\begin{align*}
		f_1(x_1',x_2,\cdots,x_n)=\dfrac{1}{2}\left|\left[0,\dfrac{1}{2}\right]\right|+\left|\left[\dfrac{1}{2},1\right]\right|
		=\dfrac{1}{4}+\dfrac{1}{2}
		>\dfrac{1}{2}
		\geq\dfrac{1}{n}=f_1({\bm x}).
		\end{align*}
		
	\end{enumerate}
	From (i)--(iii), if $x_1=0$, we have proved that $ {\bm x} $ is not equilibrium. 
	Furthermore, from the symmetry, if $x_n=1$, we can similarly prove that $ {\bm x} $ is not equilibrium.\\
	\\
	{\it Proof of \eqref{eqn:pro2}}	\\
	We prove that $ {\bm x} $ is not equilibrium, provided that $i\;(3\leq i\leq n)$ vendors occupy at a point. We assume that $x_l<x_{l+1}=\cdots = x_k = \cdots = x_{l+i}<x_{l+i+1}\quad(l\geq0,\;l+i+1\leq n+1,\;3\leq i\leq n)$.
	Here we recall that we set $x_0=0$ and $x_{n+1}=1$. Therefore there exist $x_l$ and $x_{l+i+1}$ at least one respectively. We notice 
	that $x_l<x_k<x_{l+i+1}$ and there exists no vendor on $(x_l,x_k)$ and $(x_k,x_{l+i+1})$. Dividing this 
	proof into four cases, we prove
	\eqref{eqn:pro2}.
	\begin{enumerate}
		\item $x_l\neq x_0$ and $x_{l+i+1}\neq x_{n+1}$
	In this case, if $|I_l|\geq|I_{l+i}|$, we notice that $f_k({\bm x})=\dfrac{1}{2i}(|I_l|+|I_{l+i}|)\leq \dfrac{1}{6}(|I_l|+|I_{l+i}|)=\dfrac{1}{6}(|I_l|+|I_l|)=\dfrac{1}{3}|I_l|$.
		Setting $x_k'\in{(x_l,x_{l+1})}$, we have
		$f_k(x_1,\cdots,x_{k-1},x_k',x_{k+1},\cdots,x_n)=\dfrac{1}{2}(|[x_l,x_k']|+|[x_k',x_k]|)
		=\dfrac{1}{2}|I_l|>f_k({\bm x})$.

		For the other case $|I_l|<|I_{l+i}|$, from the symmetry, we can similarly show that 
		there exists $x_k'$ such that $f_k(x_1,\cdots,x_{k-1},x_k',x_{k+1},\cdots,x_n)>f_k({\bm x})$.
		Thus, if $x_l\neq x_0$ and $x_{l+i+1}\neq x_{n+1}$, we can prove that $ {\bm x} $ is not equilibrium.
		
		\item $x_l=x_0$ and $x_{l+i+1}\neq x_{n+1}$
		

		In this case, if $|I_l|\geq|I_{l+i}|$, we find that 
		$f_k({\bm x})=\dfrac{1}{i}(|I_l|+\dfrac{1}{2}|I_{l+i}|)\leq\dfrac{1}{3}(|I_l|+\dfrac{1}{2}|I_{l+i}|)\leq\dfrac{1}{3}(|I_l|+\dfrac{1}{2}|I_l|)=\dfrac{1}{2}|I_l|$. Setting $x_k'=\dfrac{2}{3}|I_l|$, we have 
        \begin{align*}
		f_k(x_1,\cdots,x_{k-1},x_k',x_{k+1},\cdots,x_n)&=|[0,x_k']|+\dfrac{1}{2}|[x_k',x_k]|
		>|[0,x_k']|\\
		&=\dfrac{2}{3}|I_l|\geq f_k({\bm x}).
		\end{align*}
	For the other case $|I_l|<|I_{l+i}|$, from the symmetry, we can similarly show that 
	there exists $x_k'$ such that $f_k(x_1,\cdots,x_{k-1},x_k',x_{k+1},\cdots,x_n)>f_k({\bm x})$.
	Thus, if $x_l=x_0$ and $x_{l+i+1}\neq x_{n+1}$, we have showed that $ {\bm x} $ is not equilibrium. 
	\item $x_l\neq x_0$ and $x_{l+i+1}=x_{n+1}$

	From the symmetry of $(\mathrm{ii})$, there exists $x_k'$ satisfying\\ $f_k(x_1,\cdots,x_{k-1},x_k',x_{k+1},\cdots,x_n)>f_k({\bm x})$.
	Thus if $x_l\neq x_0$ and $x_{l+i+1}=x_{n+1}$, we have showed that $ {\bm x} $ is not equilibrium.
	
	\item $x_l=x_0$ and $x_{l+i+1}=x_{n+1}$
	

    In this case, we find $f_k({\bm x})=\dfrac{1}{n}$. If $x_k\geq\dfrac{1}{2}$, for $x_k'=\dfrac{2}{5}$, we have 
		\begin{align*}
		f_k(x_1,\cdots,x_{k-1},x_k',x_{k+1},\cdots,x_n)&=|[0,x_k']|+\dfrac{1}{2}|[x_k',x_k]|
		>|[0,x_k']|\\
		&=\dfrac{2}{5}\geq\dfrac{1}{n}=f_k({\bm x}).
		\end{align*}
	
	For the other case $x_k<\dfrac{1}{2}$, from the symmetry, we can similarly show that there exists $x_k'$ such that $f_k(x_1,\cdots,x_{k-1},x_k',x_{k+1},\cdots,x_n)>f_k({\bm x})$.
	Thus, if $x_l=x_0$ and $x_{l+i+1}=x_{n+1}$, we have showed that $ {\bm x} $ is not equilibrium.
\end{enumerate}
	From (i)--(iv), we can complete the proof of \eqref{eqn:pro2}.
\\
{\it Proof of \eqref{eqn:pro3}}	\\
If $x_1<x_2$, we show that $ {\bm x} $ is not equilibrium. $(x_1\neq0)$

We notice that $f_1({\bm x})=|I_0|+\dfrac{1}{2}|I_1|.$
Setting $x_1'=|I_0|+\dfrac{1}{2}|I_1|$, we have\\ 
$f_1(x_1',x_2,\cdots,x_n)=|[0,x_1']|+\dfrac{1}{2}|[x_1',x_2]|>|[0,x_1']|=|I_0|+\dfrac{1}{2}|I_1|=f_1({\bm x})$.
Thus $ {\bm x} $ is not equilibrium.
If $x_{n-1}<x_n$, we can similarly prove that $ {\bm x} $ is not equilibrium $(x_n\neq1)$.
\end{proof}


Finally, we compare a location after the movement of a vendor with the original one. To do this, we introduce the following notation.

For a given location  of $n$ vendors ${\bm x}=(x_1,x_2,x_3,\ldots,x_n)$ with $x_1\leq x_2\leq \cdots\leq x_n$, we move vendor $k$ from $x_k$ to a point in 
$A\subset [0,1]$. We denote the resultant location by $x_k\rightarrow A$. We notice that $x_k\rightarrow A$ represents the following vector
\begin{align}
(x_1,x_2,,\ldots,x_{k-1},x'_k,x_{k+1},\ldots,x_n)
,
\end{align}
where $x'_k$ is a location of vendor $k$ after movement and $x'_k\in A$.

Then we have the following lemmas. Since their proofs are a little complicated, they are postponed to Appendix.
\begin{lemma}\label{lem:2}
	If a location of $n$ vendors ${\bm x}=(x_1,x_2,x_3,\ldots,x_n)$ satisfies $\eqref{eqn:main1}$ and $\eqref{eqn:main2}$, $|I_0|\leq f_k({\bm x})\quad(1\leq k\leq n)$.
\end{lemma}

\begin{lemma}\label{lem:3}
If a location of $n$ vendors ${\bm x}=(x_1,x_2,x_3,\ldots,x_n)$ satisfies $\eqref{eqn:main1}$ and $\eqref{eqn:main2}$, $f_k(x_k\rightarrow[0,1])\leq f_k({\bm x})\quad(1\leq k\leq n)$.
\end{lemma}

\section{Proof of Theorem \ref{thm:main}} We are now position to prove our main theorem. 
\section*{Proof of Theorem \ref{thm:main} (i)}
We prove that if ${\bm x}=(x_1,x_2)=\left(\dfrac{1}{2},\dfrac{1}{2}\right)$, 
${\bm x}$ is equilibrium.

\begin{proof}
	
\begin{enumerate}
	\renewcommand{\labelenumi}{(\roman{enumi})}

	\item We consider vendor $1$. We then notice that $f_1({\bm x})=\dfrac{1}{2}$. For any $x_1'\in{\left[0,\dfrac{1}{2}\right)}$, we have 
		\begin{align*}
		f_1(x_1',x_2)=|[0,x_1']|+\dfrac{1}{2}|[x_1',x_2]|
		<|[0,x_1']|+|[x_1',x_2]|
		=|[0,x_2]|
		=\dfrac{1}{2}=f_1({\bm x}).
		\end{align*}
		For any $x_1'\in{\left(\dfrac{1}{2},1\right]}$,
		from the symmetry of (a), we can similarly prove that $f_1(x_1',x_2)<f_1({\bm x})$.	
Therefore, for any $x_1'\in{[0,1]}$, we have $f_1({\bm x})\geq f_1(x_1',x_2)$.
	\item Next, we consider vendor $2$. For any $x_2'\in{[0,1]}$, we find that $f_2({\bm x})\geq f_2(x_1,x_2')$ in a similar manner to (i).
\end{enumerate}
From $(\mathrm{i})$ and $(\mathrm{ii})$, we have showed that $\left(\dfrac{1}{2},\dfrac{1}{2}\right)$ is equilibrium. \end{proof}

Next, we show that ${\bm x}=\left(\dfrac{1}{2},\dfrac{1}{2}\right)$ is a necessary condition for equilibrium. Therefore, we prove that if ${\bm x}\ne\left(\dfrac{1}{2},\dfrac{1}{2}\right)$, then ${\bm x}$ is not equilibrium.	
\begin{proof}
From \eqref{eqn:pro1}, when $x_1=0$ or $x_2=1$, ${\bm x}$ is not equilibrium. Therefore, 
we treat with the case where $x_1\neq0$ or $x_2\neq1$.
\begin{enumerate}
	\renewcommand{\labelenumi}{(\roman{enumi})}
	\item $x_1\neq x_2$\\
	From \eqref{eqn:pro3}, $ {\bm x} $ is not equilibrium in this case.
	
	\item $x_1=x_2$
	
We first notice that $f_1({\bm x})=\dfrac{1}{2}$ in this case.
	If $x_1>\dfrac{1}{2}$, setting $x_1'=\dfrac{1}{2}$, we obtain 
		\begin{align*}
		f_1(x_1',x_2)=|[0,x_1']|+\dfrac{1}{2}|[x_1',x_2]|
		>|[0,x_1']|
		=\dfrac{1}{2}=f_1({\bm x}).
		\end{align*}
		If $x_1<\dfrac{1}{2}$,		
		from symmetry, we can show that there exits $x_1'\in{[0,1]}$ such that $f_1(x_1',x_2)>f_1({\bm x})$.

\end{enumerate}
From the above, we have proved that $\left(\dfrac{1}{2},\dfrac{1}{2}\right)$ is a necessary condition for equilibrium.
\end{proof}

\section*{Proof of Theorem \ref{thm:main} (ii)}
	If $n=3$, \eqref{eqn:pro2} contradicts \eqref{eqn:pro3}. Therefore, we conclude that 
	there exists no equilibrium in this case.

\section*{Proof of Theorem \ref{thm:main} (iii)}
Finally, we are concerned with the case where $n\geq4$.
\begin{proof}
	First, it follows from Lemma \ref{lem:3} that \eqref{eqn:main1} and \eqref{eqn:main2} is a sufficient condition for equilibrium.

	Next, we show that \eqref{eqn:main1}--\eqref{eqn:main2} is a necessary condition for equilibrium. Therefore, we prove that if \eqref{eqn:main1}--\eqref{eqn:main2} do not hold, then ${\bm x}$ is not equilibrium.
	Observing Proposition \ref{pro:1}, we do not have to treat with the case where $x_1=0$ or $x_n=1$ or more than $2$ vendors occupy a location. From this reason, we assume that $|I_0|>0$. 
	We prove the following cases. 
	\begin{enumerate}
		\renewcommand{\labelenumi}{(\roman{enumi})}
		\item Condition \eqref{eqn:main1} does not hold.
		\begin{enumerate}
			\renewcommand{\labelenumii}{(\alph{enumii})}
			\item $|I_1|\neq0\quad($resp. $|I_{n-1}|\neq0)$.
			\item $|I_1|=|I_{n-1}|=0$ and $|I_0|:|I_2|\neq1:2\quad($resp. $|I_1|=|I_{n-1}|=0$ and $|I_{n-2}|:|I_n|\neq2:1)$.
			\item $|I_1|=|I_{n-1}|=0$ and $|I_0|:|I_2|=1:2$ and $|I_{n-2}|:|I_n|=2:1$ and $|I_0|\neq|I_n|$.
		\end{enumerate}
		\item Condition \eqref{eqn:main1} holds and condition \eqref{eqn:main2} does not hold.
		\begin{enumerate}
			\renewcommand{\labelenumii}{(\alph{enumii})}
			\item Condition \eqref{eqn:main1} holds and $|I_j|>2|I_0|$.
			\item Condition \eqref{eqn:main1} holds and $|I_j|\leq2|I_0|$ and $2|I_0|>|I_k|+|I_{k+1}|$.
		\end{enumerate}
	\end{enumerate}
	Dividing this proof into the above cases, we prove our main theorem.
	
	\begin{enumerate}
		\renewcommand{\labelenumi}{(\roman{enumi})}
		\item 
		\begin{enumerate}
			\renewcommand{\labelenumii}{(\alph{enumii})}
			\item $|I_1|\neq0$\quad$($resp. $|I_{n-1}|\neq0)$

			From \eqref{eqn:pro3} and $x_1<x_2$ $($resp. $x_{n-1}<x_n)$, $ {\bm x} $ is not equilibrium.		
			
			\item $|I_1|=|I_{n-1}|=0$ and $|I_0|:|I_2|\neq1:2\;($resp. $|I_1|=|I_{n-1}|=0$ and $|I_{n-2}|:|I_n|\neq2:1$)
			\begin{enumerate}
				\renewcommand{\labelenumiii}{(\arabic{enumiii})}
				\item $|I_0|:|I_2|=|I_0|:(2|I_0|+\delta)$\quad$(\delta>0)$
				

				We notice that $f_1({\bm x})=\dfrac{1}{2}\left\{|I_0|+\dfrac{1}{2}(2|I_0|+\delta)\right\}=|I_0|+\dfrac{1}{4}\delta$. Setting $x_1'\in{(x_2,x_3)}$, we have 
				\begin{align*}
				f_1(x_1',x_2,\cdots,x_n)=\dfrac{1}{2}(2|I_0|+\delta)
				=|I_0|+\dfrac{1}{2}\delta
				>|I_0|+\dfrac{1}{4}\delta=f_1({\bm x}).
				\end{align*}
				
				\item $|I_0|:|I_2|=|I_0|:(2|I_0|-\delta)$\quad$(\delta>0)$
				

				We notice that $f_1({\bm x})=\dfrac{1}{2}\left\{|I_0|+\dfrac{1}{2}\left(2|I_0|-\delta\right)\right\}=|I_0|-\dfrac{1}{4}\delta$. Setting $x_1'=|I_0|-\dfrac{1}{4}\delta$, we have 
				\begin{align*}
				f_1(x_1',x_2,\cdots,x_n)&=|[0,x_1']|+\dfrac{1}{2}|[x_1',x_2]|>|[0,x_1']|\\
				&=|I_0|-\dfrac{1}{4}\delta=f_1({\bm x}).
				\end{align*}
			\end{enumerate}
			From the symmetry of (1)--(2), we can similarly show in the case where $|I_{n-2}|:|I_n|\neq2:1$. Thus, ${\bm x}$ is not equilibrium in the case of (b).
			
			\item $|I_1|=|I_{n-1}|=0$ and $|I_0|:|I_2|=1:2$ and $|I_{n-2}|:|I_n|=2:1$ and $|I_0|\neq|I_n|$
			\begin{enumerate}
				\renewcommand{\labelenumiii}{(\arabic{enumiii})}
				\item $|I_0|<|I_n|$
				

				We notice that $f_1({\bm x})=\dfrac{1}{2}\left(|I_0|+\dfrac{1}{2}\cdot2|I_0|\right)=|I_0|$. Setting $x_1'=1-|I_0|$, we have 
				\begin{align*}
				f_1(x_1',x_2,\cdots,x_n)=\dfrac{1}{2}|[x_n,x_1']|+|[x_1',1]|
				>|[x_1',1]|
				=|I_0|=f_1({\bm x}).
				\end{align*}
				
				\item $|I_0|>|I_n|$\\
				From the symmetry, we can similarly show that there exists $x_n'$ such that $f_n(x_1,\cdots,x_{n-1},x_n')>f_n({\bm x})$.
			\end{enumerate}
			From (1)--(2), ${\bm x}$ is not equilibrium in the case of (c).
		\end{enumerate}
		
		\item Condition \eqref{eqn:main1} holds and condition \eqref{eqn:main2} does not hold.
		\begin{enumerate}
			\renewcommand{\labelenumii}{(\alph{enumii})}
			\item \eqref{eqn:main1} holds and $|I_j|>2|I_0|$

			We notice that $f_1({\bm x})=\dfrac{1}{2}\left(|I_0|+\dfrac{1}{2}\cdot2|I_0|\right)=|I_0|$. Setting $x_1'\in{(x_j,x_{j+1})}$, we have \\
			$f_1(x_1',x_2,\cdots,x_n)=\dfrac{1}{2}|I_j|>\dfrac{1}{2}\cdot2|I_0|=|I_0|=f_1({\bm x})$.\\
			Thus, ${\bm x}$ is not equilibrium in this case.
			\item \eqref{eqn:main1} holds and $|I_j|\leq2|I_0|$ and $2|I_0|>|I_k|+|I_{k+1}|$

			When $k=1,2,n-1,n$, we notice that $|I_1|+|I_2|=2|I_0|$ and $|I_2|+|I_3|\geq2|I_0|$. Thus, we devote to considering $3\leq k\leq n-2$.
			In view of \eqref{eqn:pro2}, we divide (b) into  
			the following three parts.
		
			\begin{enumerate}
				\renewcommand{\labelenumiii}{(\arabic{enumiii})}
				\item $x_{k-1}\neq x_k$ and $x_k\neq x_{k+1}$
				

				We notice that $f_k({\bm x})=\dfrac{1}{2}(|I_k|+|I_{k+1}|)<\dfrac{1}{2}\cdot2|I_0|=|I_0|$. Setting $x_k'\in{(x_2,x_3)}$, we have $f_k(x_1,\cdots,x_{k-1},x_k',x_{k+1},\cdots,x_n)=\dfrac{1}{2}\cdot2|I_0|=|I_0|>f_k({\bm x})$.
				Thus ${\bm x}$ is not equilibrium in this case.
				
				\item $x_{k-1}=x_k$\quad$(x_{k-2}\neq x_{k-1}$ and $x_k\neq x_{k+1})$
				
		
				From $|I_k|+|I_{k+1}|<2|I_0|$, we find $|I_k|<2|I_0|$. From $|I_j|\leq 2|I_0|$, we notice that 
				$|I_{k-2}|\leq2|I_0|$. It follows that $f_k({\bm x})=\dfrac{1}{4}(|I_{k-2}|+|I_k|)<\dfrac{1}{4}(2|I_0|+2|I_0|)=|I_0|$.
				Therefore, for $x_k'\in{(x_2,x_3)}$, we have \linebreak
				$f_k(x_1,\cdots,x_{k-1},x_k',x_{k+1},\cdots,x_n)=\dfrac{1}{2}\cdot2|I_0|=|I_0|>f_k({\bm x})$.
				This means that ${\bm x}$ is not equilibrium in this case.
				\item $x_k=x_{k+1}$\quad$(x_{k-1}\neq x_k$ and $x_{k+1}\neq x_{k+2})$
				
We can prove this case in a similar manner to (2).
			\end{enumerate}		
		\end{enumerate}
	\end{enumerate}

	From (i)--(ii), we have showed that if condition \eqref{eqn:main1} or condition \eqref{eqn:main2} do not hold, then ${\bm x}$ is not equilibrium.

	We can complete the proof of Theorem \ref{thm:main}.
\end{proof}

\appendix

\section{Proof of Lemma \ref{lem:2}}
\begin{proof}
We estimate the profit of each vendor, $f_k({\bm x})\quad(k=1,2,\ldots,n)$.
\begin{enumerate}
	\renewcommand{\labelenumi}{(\roman{enumi})}
	\item Vendor $1$
	

	We have $f_1({\bm x})=\dfrac{1}{2}(|I_0|+\dfrac{1}{2}\cdot2|I_0|)=|I_0|$.
	\item Vendor $2$
	
	We can similarly deduce $f_2({\bm x})=|I_0|$.
	\item Vendor $n-1$ and $n$

	From the symmetry with vendors $1$ and $2$, we have $f_{n-1}({\bm x})=f_n({\bm x})=|I_0|$.
	\item Vendor $k$ $(3\leq k\leq n-2)$
	\begin{enumerate}
		\renewcommand{\labelenumii}{(\alph{enumii})}
		\item $x_{k-1}\neq x_k$ and $x_k\neq x_{k+1}$
	

		We have $f_k({\bm x})=\dfrac{1}{2}(|I_{k-1}|+|I_k|)\geq\dfrac{1}{2}\cdot2|I_0|=|I_0|$.
		\item $x_{k-1}=x_k,x_k\neq x_{k+1}$\quad$(k\neq3)$
	

		From \eqref{eqn:main2}, we notice that $x_{k-2}\neq x_{k-1}$. 
		On the other hand, from $|I_{k-1}|=0$ and \eqref{eqn:main2}, we have $|I_{k-2}|\leq2|I_0|$ and $2|I_0|\leq|I_{k-2}|$. Thus we have $|I_{k-2}|=2|I_0|$. Similarly, we have $|I_k|=2|I_0|$.
		As a consequence, we have 
		$f_k({\bm x})=\dfrac{1}{4}(|I_{k-2}|+|I_k|)=\dfrac{1}{4}\cdot4|I_0|=|I_0|$.\\

		\item $x_{k-1}\neq x_k,x_k=x_{k+1}$\quad$(k\neq n-2)$
	

		\vspace*{3ex}
		We can show that $f_k({\bm x})=|I_0|$ in a similar manner to (b).
		
	\end{enumerate}
\end{enumerate}

Combining (i)--(iv), we obtain 
\[
\begin{cases}
f_1({\bm x})=f_2({\bm x})=f_{n-1}({\bm x})=f_n({\bm x})=|I_0|.\\
f_k({\bm x})=|I_0|\quad\text{if}\; x_{k-1}=x_k\;\text{or}\;x_k=x_{k+1}\quad
(3\leq k\leq n-3).\\
f_k({\bm x})\geq|I_0|\quad\text{if}\;x_{k-1}\neq x_k\;\text{and}\;x_k\neq x_{k+1}\quad(3\leq k\leq n-3).
\end{cases}
\]
Thus, for ${\bm x}$ satisfying $\eqref{eqn:main1}$ and $\eqref{eqn:main2}$ and any $k$ $(1\leq k \leq n)$, we have showed that $f_k({\bm x})\geq|I_0|$.

\end{proof}

\section{Proof of Lemma \ref{lem:3}}
\begin{proof}
	Dividing the proof into three parts, we prove this lemma.		
	\begin{enumerate}
		\item $f_k(x_k\rightarrow(x_l,x_{l+1}))\leq f_k({\bm x})\quad(x_l<x_{l+1}$ and $l\neq k-1,k$ and $0\leq l\leq n)$. 
		\item $f_k(x_k\rightarrow(x_{k-1},x_{k+1}))\leq f_k({\bm x})$\quad$(1\leq k\leq n)$.
		\item $f_k(x_k\rightarrow\{x_l\})\leq f_k({\bm x})$\quad$(1\leq k\leq n$ and $0\leq l\leq n+1)$.
	\end{enumerate}
	{\it Proof of $(\mathrm{i})$}
	\begin{enumerate}
		\renewcommand{\labelenumi}{(\alph{enumi})}
		\item $l\neq0$ and $l\neq n$

		We have $f_k(x_k\rightarrow(x_l,x_{l+1}))=\dfrac{1}{2}|I_l|\leq\dfrac{2|I_0|}{2}=|I_0|\leq f_k({\bm x})$.
		\item $l=0$

		We have $f_k(x_k\rightarrow(x_0,x_1))<|I_0|\leq f_k({\bm x})$.
		\item $l=n$
		
		We have $f_k(x_k\rightarrow(x_n,x_{n+1}))<|I_0|\leq f_k({\bm x})$.
	\end{enumerate}

	{\it Proof of $(\mathrm{ii})$}
	\begin{enumerate}
		\renewcommand{\labelenumi}{(\alph{enumi})}
		\item $k\neq1$ and $k\neq n$
		\begin{enumerate}	\renewcommand{\labelenumii}{(\arabic{enumii})}
			\item If $x_{k-1}\neq x_k$ and $x_k\neq x_{k+1}$, we have $f_k(x_k\rightarrow(x_{k-1},x_{k+1}))=f_k({\bm x})$;
			\item If $x_{k-1}=x_k\;(x_k\neq x_{k+1})$, since $|I_k|=2|I_0|$ from \eqref{eqn:main2} and $|I_{k-1}|=0$, we have $f_k(x_k\rightarrow(x_{k-1},x_{k+1}))=\dfrac{1}{2}\cdot2|I_0|=|I_0|\leq f_k({\bm x})$;
			\item If $x_k=x_{k+1}\quad(x_{k-1}\neq x_k)$, we can deduce  $f_k(x_k\rightarrow(x_{k-1},x_{k+1}))=|I_0|\leq f_k({\bm x})$ in a similar manner to (2).
		\end{enumerate}

		\item $k=1$
		
		We have $f_1(x_1\rightarrow(x_0,x_2))<|I_0|\leq f_1({\bm x})$.
		\item $k=n$

		We have $f_n(x_n\rightarrow(x_{n-1},x_{n+1}))<|I_0|\leq f_n({\bm x})$.
	\end{enumerate}

	{\it Proof of $(\mathrm{iii})$}
	\begin{enumerate}
		\renewcommand{\labelenumi}{(\alph{enumi})}
		\item $l\neq0,1,2,n-1,n,n+1$
			\begin{enumerate}
			\renewcommand{\labelenumii}{(\arabic{enumii})}
			\item $l=k$   
			
			It clearly holds that $f_k(x_k\rightarrow\{x_l\})=f_k({\bm x})$.
			\item $l=k-1,k+1$ and there exists another vendor at $x_k$ except for vendor $k$, or $l\ne k-1,k,k+1$
			
			We deduce from $\eqref{eqn:main2}_1$ that $f_k(x_k\rightarrow\{x_l\})\leq\dfrac{1}{4}\cdot4|I_0|=|I_0|\leq f_k({\bm x})$.
			\item $l=k-1,k+1$ and there exists no vendor at $x_k$ except for vendor $k$.
			
			
			We deduce from $\eqref{eqn:main2}_1$ that $f_k(x_k\rightarrow\{x_l\})\leq\dfrac{1}{4}(2|I_0|+2f_k({\bm x}))\leq\dfrac{1}{4}(2f_k({\bm x})+2f_k({\bm x}))=f_k({\bm x})$.
			
	    	\end{enumerate}
    	\item $l=0\;($resp.$\;l=n+1)$
    	
    	We have $f_k(x_k\rightarrow\{x_0\})=\dfrac{1}{2}|I_0|<|I_0|\leq f_k({\bm x}).$\\
    	(resp.$\;f_k(x_k\rightarrow\{x_{n+1}\})=\dfrac{1}{2}|I_0|<|I_0|\leq f_k({\bm x}))$
    	\item $l=1,2\;($resp.$\;l=n-1,n)$
    		\begin{enumerate}
    		\renewcommand{\labelenumii}{(\arabic{enumii})}
    		\item $k=3\;($resp.$\;k=n-2)$
    		
    		We have $f_3(x_3\rightarrow\{x_l\})=\dfrac{1}{3}(|I_0|+\dfrac{1}{2}\left(|I_2|+|I_3|)\right)\leq\dfrac{1}{3}(|I_0|+\dfrac{1}{2}(2|I_0|+2|I_0|))=|I_0|\leq f_3({\bm x}).$\\
    		(resp.$\;f_{n-2}(x_{n-2}\rightarrow\{x_l\})=\dfrac{1}{3}(|I_n|+\dfrac{1}{2}(|I_{n-3}|+|I_{n-2}|)\leq f_{n-2}({\bm x}))$
    		\item $k\neq3\;($resp.$\;k\neq n-2)$
    		
    		We have $f_k(x_k\rightarrow\{x_l\})\leq\dfrac{1}{2}(|I_0|+\dfrac{1}{2}|I_2|)=\dfrac{1}{2}(|I_0|+\dfrac{1}{2}\cdot2|I_0|)=|I_0|\leq f_k({\bm x}).$\\
    		(resp.$\;f_k(x_k\rightarrow\{x_l\})\leq\dfrac{1}{2}(|I_0|+\dfrac{1}{2}|I_{n-2}|)\leq f_k({\bm x}))$
         	\end{enumerate}
	\end{enumerate}
		We can complete the proof of Lemma \ref{lem:3}.
\end{proof}

\section*{Acknowledgements} 
The authors would like to thank Prof. Suzuki for his kind help and comments.

\end{document}